\documentclass[12pt,final]{l4dc2020} 
\title[Learning the Globally Optimal Distributed LQ Regulator]{Learning the Globally Optimal Distributed LQ Regulator}
\usepackage{times}
	 \author{\Name{Luca Furieri} \Email{furieril@control.ee.ethz.ch}\\  \addr  Automatic Control Laboratory, ETH Zurich, Switzerland \AND \Name{Yang Zheng} \Email{zhengy@g.harvard.edu} \\  \addr School of Engineering and Applied Sciences, Harvard University, USA   \AND \Name{Maryam Kamgarpour}  \thanks{This research was gratefully funded by the European Union ERC Starting Grant CONENE.} \Email{mkamgar@control.ee.ethz.ch}\\ \addr  Automatic Control Laboratory, ETH Zurich, Switzerland 
}

\usepackage{url}

\usepackage{graphicx}
\usepackage{caption}
\usepackage{float}
\usepackage{booktabs}
\usepackage{multirow}
\usepackage{xcolor}
\usepackage{graphics}
\usepackage{caption}
\usepackage{bm}
\usepackage{epstopdf}
\usepackage{amsmath}
	\usepackage{amssymb}
	 \usepackage{enumerate}
	 \usepackage{balance}
	 \usepackage{empheq}
	 \usepackage{mathtools}

	 \usepackage{algorithm}
\usepackage{algorithmic,bm,multirow,array}

\newcommand{\tr}{{{\mathsf T}}}
\DeclareMathOperator*{\minimize}{minimize}

\newcommand{\preprintswitch}[2]{#2} 

\usepackage{hyperref}
\hypersetup{
    colorlinks=true,
    linkcolor=black,
    filecolor=black,
    urlcolor=black,
}

\mathchardef\Re="023C
\mathchardef\Im="023D

\newcommand{\norm}[1]{\left\lVert#1\right\rVert}

\begin{document}
\maketitle


	%

\maketitle

\begin{abstract} 
We study model-free learning methods for the output-feedback Linear Quadratic (LQ) control problem in finite-horizon subject to subspace constraints on the control policy. Subspace constraints naturally arise in the field of distributed control and present a significant challenge in the sense that standard model-based optimization and learning leads to intractable numerical programs in general.  Building upon recent results in zeroth-order optimization, we establish model-free sample-complexity bounds for the class of distributed LQ problems where a local gradient dominance constant exists on any sublevel set of the cost function. 
We prove that a fundamental class of distributed control problems{---}commonly referred to as Quadratically Invariant (QI) problems{---}as well as others possess this property. To the best of our knowledge, our result is the first sample-complexity bound guarantee on learning globally optimal distributed output-feedback 
control policies. 

\end{abstract}

\section{Introduction}
 Recent years have witnessed significant attention and progress in  controlling unknown dynamical systems solely based on 
 system trajectory observations. This shift from classical control approaches to data-driven ones is motivated  by the ever increasing complexity of critical emerging dynamical systems, whose mathematical models 
  may be 
 unreliable or simply not available \citep{hou2013model}. When it comes to learning an optimal control policy, the available approaches can be broadly divided into two categories. The first class of methods is denoted as \emph{model-based}, where the historical system data is exploited to build an approximation of the nominal system and classical optimal robust control is then used on this system approximation. The second class of methods is denoted as \emph{model-free}, where reinforcement learning 
 is used to directly learn an optimal control policy based on the observed costs, without explicitly constructing a model for the system.

 Model-free approaches tend to require more samples to achieve a policy of equivalent accuracy~\citep{tu2018gap}, but are inherently unaffected by the potential challenges of designing an optimal controller. Indeed, in large-scale dynamical systems, the control policy is often required to be \emph{distributed}, in the sense that different controllers can only base their control policy on partial sensor measurements due to limited sensing capabilities, geographic distance or privacy concerns. Given such limitations, it has been known  that the corresponding optimization problems are NP-hard in general~\citep{papadimitriou1986intractable,blondel2000survey,Witsenhausen}. Often, one can only derive a tractable approximation using convex relaxations (e.g. \cite{fazelnia2016convex}) or restrictions (e.g. \cite{furieri2019sparsity}). The difficulties in solving model-based optimal control for large-scale systems motivate us to bypass numerical programming altogether and study the properties of 
 model-free methods for distributed control. 

 
For Linear  Quadratic (LQ) control problems in infinite-horizon without additional constraints, the optimal policy  can be derived with dynamic programming by solving a Riccati equation. 
For distributed control tasks, the optimal policy might not  be linear in general \citep{Witsenhausen} and even in those cases where an explicit solution can be computed (see e.g., \cite{lamperski2015cal} and references within), the optimal controller requires several internal states and might admit a rather complicated formulation. Furthermore, when designing a \emph{static} distributed controller in infinite-horizon,  model-free methods are unlikely to find the globally optimal controller 
due to the feasible set being disconnected in general \citep{feng2019exponential}; for this setting, convergence to \emph{local} optima was confirmed by \cite{hassan2019data}.  

Motivated as per above, in this paper we consider model-free learning of globally optimal dynamic distributed controllers. We focus on the \emph{finite-horizon} setup, where the feasible set is naturally connected because every control policy yields a finite closed-loop cost. 
Furthermore, in this setup we can 1) encode general \emph{dynamic} time-varying linear policies in a relatively simple 
way, and 2) consider time-varying system dynamics. 

\textbf{Our contributions}\hspace{0.5cm} First, we provide a general framework for model-free learning of distributed dynamic linear policies in finite-horizon with uncertain initial state, process noise and noisy output observations. 
Second, our key contribution is to establish a property of local gradient dominance for a class of distributed control problems, including 1) all Quadratically Invariant (QI) problems \citep{rotkowitz2006characterization} and 2) some non-QI problems. This local gradient dominance property is crucial for establishing model-free sample-complexity bounds using zeroth-order optimization; we base our corresponding analysis on the recent results of \cite{malik2018derivative}, while adapting and extending relevant aspects. 

\textbf{Related work}\hspace{0.5cm} Thanks to its well-understood solution structure and its properties, the LQ problem has enjoyed significant attention in the line of work on model-based learning, originating from  classical system identification (see \cite{ljung2010perspectives} 
for a nice overview). A non-asymptotic analysis was provided by \cite{fiechter1997pac} and significantly refined by \cite{dean2017sample}, and sub-linear regret results for online model-based methods were recently obtained  by \cite{dean2018regret,abbasi2011regret,abeille2018improved}. Still assuming full sensor information, 
\cite{mania2019certainty} exploited Riccati perturbation theory to analyse the 
output-feedback case and \cite{dean2019safely}  included safety constraints on states and inputs. The literature on model-free learning has recently been attracting significant research interest starting from the works of \cite{fazel2018global} and \cite{abbasi2018regret}.  Related to our work is \cite{fazel2018global}, which  showed that for the state-feedback LQ problem without an information structure, a standard 
policy-gradient method is guaranteed to converge  to the global optimum and established sample-complexity bounds that scaled with $\tilde{\mathcal{O}}(\epsilon^{-4})$, where $\epsilon$ is the  suboptimality gap. This bound was  improved to $\tilde{\mathcal{O}}(\epsilon^{-2})$ in \cite{malik2018derivative}, 
at the expense of a constant probability of success, for a  discounted LQ cost function. Furthermore,  similar convergence properties were shown for  robust control tasks without an information structure; we refer the  reader  to \cite{gravell2019learning} for the case of multiplicative noise and to \cite{zhang2019policy} for  $\mathcal{H}_\infty$ robustness guarantees.

To the best of our knowledge, 
 global convergence for distributed control problems, where a subspace constraint is imposed on the control policy, has not been studied from a model-free perspective.
A related problem has been addressed with a model-based approach in \cite{fattahi2019efficient}, where the authors extended the method of \cite{dean2017sample} by adding subspace constraints on the \emph{closed-loop responses}. 
In general, a sparse closed-loop  response does not lead to a  sparse controller implementation that is exclusively based on measuring the outputs, and vice-versa (see~\cite{zheng2019equivalence} for details on this aspect). The resulting framework is thus not directly comparable with the one considered in this paper. We also note that the work of \cite{fattahi2019efficient} restricts the analysis to state-feedback, whereas we consider noisy output-feedback.

\label{se:introduction}

\section{Background and Problem Statement}
\label{se:preliminaries}

\textbf{Notation: }We use $\mathbb{R}$ and $\mathbb{N}$ to denote the set of real numbers and integers, respectively.  We write $M=\text{blkdg}(M_1,\ldots,M_n)$ to denote a block-diagonal matrix with $M_1,\ldots,M_n$ on its diagonal block entries. 
The Kronecker product between $M \in \mathbb{R}^{m \times n}$ and $P \in \mathbb{R}^{p \times q}$ is denoted as 
$M \otimes P \in \mathbb{R}^{mp \times nq}$. 
Given $K \in \mathbb{R}^{m \times n}$,  $\text{vec}(K) \in \mathbb{R}^{mn}$ is a column vector that stacks the columns of $K$. 
We define the inverse operator $\text{vec}^{-1}: \mathbb{R}^{mn} \rightarrow \mathbb{R}^{m \times n}$ that maps a vector into a matrix (the matrix dimension shall be clear in the context).  The Euclidean norm of a vector $v \in \mathbb{R}^n$ is denoted by $\norm{v}_2^2=v^\mathsf{T}v$  and the Frobenius norm of a matrix $M \in \mathbb{R}^{m \times n}$ is denoted by $\norm{M}_{F}^2=\text{Trace}(M^\mathsf{T}M)$. For a symmetric matrix $M$, 
we write $M \succ 0$ (resp. $M \succeq 0$) if and only if it is positive definite (resp. positive semidefinite). 
We say that $x\sim \mathcal{D}$ if the random variable  $x\in \mathbb{R}^n$ is distributed according to $\mathcal{D}$.  Given a binary matrix $X \in \{0,1\}^{m \times n}$, we define the associated \emph{sparsity subspace} as
	\begin{align*}
\text{Sparse}(X)\hspace{-0.1cm}:=\{Y \in \mathbb{R}^{m \times n}\hspace{-0.1cm} \mid  Y_{i,j}\hspace{-0.1cm}=0 ~~\text{if } X_{i,j}=0, i = 1, \ldots, m, j = 1, \ldots, n \;  \}\,.
	\end{align*}
The set $\mathbb{S}_r \subseteq \mathbb{R}^d$ denotes the shell of radius $r>0$ in $\mathbb{R}^d$, that is $\mathbb{S}_r=\{z\in \mathbb{R}^d| \norm{z}_2=r\}$. A zero block of dimension $m \times n$ is denoted as $0_{m\times n}$.

\subsection{The LQ  Optimal Control Problem Subject To Subspace Constraints}\label{se:problems}

	
	
	We consider time-varying linear systems in discrete-time
	\begin{align}
	\label{eq:sys_disc}
	x_{t+1}&=A_tx_t+B_tu_t+w_t\,, \quad y_t=C_tx_t+v_t \,,
	\end{align}
	where $x_t \in \mathbb{R}^n$ is the system state at time $t$ affected by process noise $w_t \sim \mathcal{D}_w$ with $x_0=\mu_0+\delta_0$, $\delta_0 \sim \mathcal{D}_{\delta_0}$,  $y_t\in \mathbb{R}^p$ is the  observed output at time $t$ affected by  measurement noise	$v_t\sim \mathcal{D}_v$, 
	and  $u_t \in \mathbb{R}^m$ is the control input at time $t$ to be designed. We assume that the distributions $\mathcal{D}_w,\mathcal{D}_{\delta_0}$ $\mathcal{D}_v$ are bounded, have zero mean and variances of  $\Sigma_w,\Sigma_{\delta_0},\Sigma_v \succ 0$ respectively. Boundedness of the disturbances is a reasonable assumption in physical applications and it is commonly exploited to simplify the analysis of model-free methods \citep{fazel2018global,malik2018derivative}\footnote{\cite{malik2018derivative} noted that extension to sub-Gaussian disturbances is possible; 
    we leave this case to future work.
     }.  We consider the evolution of \eqref{eq:sys_disc} in finite-horizon for $t=0,\ldots N$, where $N \in \mathbb{N}$. 
 	By defining the matrices
\begin{equation*}
\mathbf{A}=\text{blkdg}(A_0,\ldots,A_{N}), \quad  \mathbf{B}\hspace{-0.1cm}=\hspace{-0.1cm}\begin{bmatrix}
 \text{blkdg}(B_0,\ldots,B_{N{-}1})\\
 0_{n \times mN}
 \end{bmatrix}, \quad \mathbf{C}=\text{blkdg}(C_0,\ldots,C_{N})
\,,
\end{equation*}
 and  the vectors $\mathbf{x}=\begin{bmatrix}x_0^\mathsf{T}&\ldots&x_N^\mathsf{T} \end{bmatrix}^\mathsf{T}$, $\mathbf{y}=\begin{bmatrix}y_0^\mathsf{T}&\ldots&y_{N}^\mathsf{T} \end{bmatrix}^\mathsf{T}$, $\mathbf{u}=\begin{bmatrix}u_0^\mathsf{T}&\ldots&u_{N-1}^\mathsf{T}\end{bmatrix}^\mathsf{T} $, $\mathbf{w}=\begin{bmatrix}x_0^\mathsf{T}&w_0^\mathsf{T}&\ldots&w_{N-1}^\mathsf{T}\end{bmatrix}^\mathsf{T}$ and $\mathbf{v}=\begin{bmatrix}v_0^\mathsf{T}&\ldots&v_{N}^\mathsf{T}\end{bmatrix}^\mathsf{T}$, and the block-down shift matrix $$\mathbf{Z}=\begin{bmatrix}
0_{1 \times N}&0\\
I_{N}&0_{N \times 1} 
\end{bmatrix}\otimes I_n\,,$$ 
we can  write the system~\eqref{eq:sys_disc} compactly as $\mathbf{x}= \mathbf{Z}\mathbf{A}\mathbf{x}+\mathbf{Z}\mathbf{B}\mathbf{u}+\mathbf{w}$,  $\mathbf{y}=\mathbf{Cx}+\mathbf{v},$ leading to 
\begin{align}
\label{eq:system_compact}
&\mathbf{x}=\mathbf{P}_{11}\mathbf{w}+\mathbf{P}_{12}\mathbf{u}\,, \quad \mathbf{y}=\mathbf{Cx}+\mathbf{v},
\end{align}
where  $\mathbf{P}_{11}=(I-\mathbf{Z}\mathbf{A})^{-1}$ and $\mathbf{P}_{12}=(I-\mathbf{Z}\mathbf{A})^{-1}\mathbf{Z}\mathbf{B}$. In this paper, we consider linear output-feedback policies 
$u_t = K_{t,0}y_0 + K_{t,1}y_1, + \ldots, K_{t,t}y_t, t = 0, 1, \ldots, N-1$. More compactly
\begin{equation}
\label{eq:control_input_def}
\mathbf{u}=\mathbf{Ky}, \quad \mathbf{K} \in \mathcal{K}\,,
\end{equation}
where $\mathcal{K}$ is a subspace in $\mathbb{R}^{mN \times p(N+1)}$ that $1)$ ensures causality of $\mathbf{K}$ by setting to $0$ those entries that correspond to future outputs, 
$2)$ 
can enforce a time-varying spatio-temporal information structure for distributed control. The presence of these information constraints presents a significant challenge for optimal distributed control; we refer to \cite{furieri2019unified} for details. 
%


The distributed Linear Quadratic (LQ) optimal control problem in finite-horizon is:
 \begin{equation*}
\text{\textbf{Problem $LQ_\mathcal{K}$}:}\qquad \min_{\mathbf{K} \in \mathcal{K}} \quad J(\mathbf{K})\,,
 \end{equation*}
where the cost $J(\mathbf{K})$ is defined as
\begin{equation}
\label{eq:cost}
J(\mathbf{K})\hspace{-0.06cm}:=\hspace{-0.06cm}\mathbb{E}_{\mathbf{w},\mathbf{v}}\hspace{-0.06cm}\left[\sum_{t=0}^{N-1}\hspace{-0.1cm}\left(y_t^\mathsf{T}M_ty_t\hspace{-0.06cm}+\hspace{-0.06cm}u_t^\mathsf{T}R_tu_t\right)+y_N^\mathsf{T}M_N y_N\hspace{-0.06cm}\right]\hspace{-0.15cm}\,,
\end{equation}	
and  $M_t \succeq 0$ and $R_t \succ 0$ for every $t$. We denote the optimal value of problem $LQ_\mathcal{K}$ as $J^\star$. 
By rearranging~\eqref{eq:system_compact}-\eqref{eq:control_input_def}, it can be observed that $J(\mathbf{K})$ is in general a non-convex multivariate polynomial in the entries of $\mathbf{K}$; see  \preprintswitch{Appendix~A of our Arxiv report  \cite{furieri2019learning}}{Appendix~\ref{ap:propertiesJ}} for an explicit expression of $J(\mathbf{K})$ and some useful properties. Note that $LQ_\mathcal{K}$ is a constrained problem over the subspace $\mathcal{K}$; it is convenient to observe that $LQ_\mathcal{K}$ is actually equivalent to an unconstrained problem.
\begin{lemma}
Let $d\in \mathbb{N}$ be the dimension of $\mathcal{K}$, and the columns of $P \in \mathbb{R}^{mpN(N+1)\times d}$ be a basis of the subspace $\{\text{vec}(\mathbf{K})|~\forall \mathbf{K} \in \mathcal{K}\}$. Define the function $f:\mathbb{R}^d \rightarrow \mathbb{R}$ as $f(z):=J(\text{vec}^{-1}(Pz))$. Then, $LQ_\mathcal{K}$ is equivalent to  the \emph{unconstrained} problem\footnote{Throughout this paper, $J(\mathbf{K})$ is reserved for the LQ cost function in~\eqref{eq:cost} and $f(z)$ is reserved for the equivalent cost function $f(z):=J(\text{vec}^{-1}(Pz))$.}
\begin{equation}
\label{eq:problem_LQK_vectorized}
\min_{z \in \mathbb{R}^d} f(z)\,.
\end{equation}
\end{lemma}
\begin{proof}
 Since the columns of $P$ are a basis of $\mathcal{K}$, we have 1) $\forall \mathbf{K} \in \mathcal{K},~ \exists z \in \mathbb{R}^d$ such that $\text{vec}(\mathbf{K}) = Pz$ and 2) $\forall z \in \mathbb{R}^d, ~\text{vec}^{-1}(Pz) \in \mathcal{K}$. Hence, \eqref{eq:problem_LQK_vectorized} is equivalent to $LQ_\mathcal{K}$.
\end{proof}

%

The function $f(z)$ is generally a non-convex multivariate polynomial in $z \in \mathbb{R}^d$ which may possess multiple local-minima, thus preventing global convergence of model-free 
algorithms. Furthermore, as opposed to the standard LQ problem without subspace constraints, one cannot in general exploit a tractable reformulation or Riccati-based solutions and apply model-based learning as per e.g. \cite{dean2017sample,mania2019certainty}. Fortunately,  $f(z)$  admits a unique global minimum if it is \emph{gradient dominated} 
\emph{i.e.}, $\mu(f(z)-J^\star)\leq  \norm{\nabla f(z)}_2^2\,, ~ \forall z \in \mathbb{R}^d$ for some $\mu>0$ \citep{karimi2016linear}. 
Gradient dominance has been proved for the standard LQ problem in infinite horizon without subspace constraints \citep{fazel2018global,gravell2019learning}.
Inspired by these recent results, we explore conditions under which $LQ_\mathcal{K}$ admits a gradient dominance constant, to be exploited for model-free learning of globally optimal distributed controllers. 

\section{Local Gradient Dominance for QI Problems and Beyond}
\label{sub:QI}

It is well-known since the work of \cite{rotkowitz2006characterization} that problem $LQ_\mathcal{K}$ can be equivalently transformed into a  strongly convex program if and only if QI holds, that is
\begin{equation}
\label{eq:QI}
\mathbf{KCP}_{12}\mathbf{K} \in \mathcal{K},\quad \forall \mathbf{K} \in \mathcal{K}\,. 
\end{equation}
We refer to \preprintswitch{Appendix~B of our Arxiv report \cite{furieri2019learning}}{Appendix~\ref{ap:PL}} for a detailed discussion of the QI property. In model-free learning, one directly investigates whether $LQ_\mathcal{K}$ possesses favourable properties for convergence (such as gradient dominance) rather than convexifying through a system-dependent change of variables.  Our main contribution is to prove a \emph{local} gradient dominance property for 1) the class of all QI instances of $LQ_\mathcal{K}$ 2) other non-QI instances of $LQ_\mathcal{K}$. 

\begin{theorem}
\label{th:PL2}
Let $\mathcal{K}$ be QI with respect to $\mathbf{CP}_{12}$, i.e.,~\eqref{eq:QI} holds. For any $\delta>0$ and initial value $z_0 \in \mathbb{R}^d$, define the sublevel set $\mathcal{G}_{10\delta^{-1}}=\{z \in \mathbb{R}^d \mid~f(z)-J^\star\leq 10\delta^{-1} \Delta_0\}$, where $\Delta_0 := f(z_0)-J^\star$ is the initial optimality gap.  Then, the following statements hold.
\begin{enumerate}
    \item $\mathcal{G}_{10\delta^{-1}}$ is compact.
    
\item $f(z)$ has a unique stationary point.
   
    \item $f(z)$ admits a local gradient dominance constant $\mu_\delta > 0$ over 
$\mathcal{G}_{10\delta^{-1}}$, that is
\begin{equation}
\label{eq:LGD}
    \mu_\delta(f(z)-J^\star)\leq  \norm{\nabla f(z)}_2^2\,, ~ \forall z \in \mathcal{G}_{10\delta^{-1}}\,.
\end{equation}

\end{enumerate}

\end{theorem}
The proof of Theorem~\ref{th:PL2} is reported in \preprintswitch{Appendix~B of our Arxiv report \cite{furieri2019learning}}{Appendix~\ref{ap:PL}}. In other words, QI guarantees existence of a gradient dominance constant $\mu_\delta$ which is ``global'' on $\mathcal{G}_{10\delta^{-1}}$, for any $\delta>0$. By inspection of (\ref{eq:LGD}), for every $\delta>0$, the only  stationary point contained in $\mathcal{G}_{10\delta^{-1}}$ is the global optimum, since whenever $\nabla f(z) = 0$, we have $f(z) = J^\star$.

We remark that the property \eqref{eq:LGD} is weaker than the more common global gradient dominance; we present a simple instance of $LQ_\mathcal{K}$ satisfying \eqref{eq:LGD} in \preprintswitch{Appendix~B of our Arxiv report \cite{furieri2019learning}}{Appendix~\ref{ap:PL}}. We will show in Section~\ref{sec:learning} that \eqref{eq:LGD} is sufficient for global convergence of model-free algorithms. Furthermore, diverse classes of non-QI  $LQ_\mathcal{K}$ that yet are convex in $\mathbf{K}$ have been found in \cite{lessard2010internal,shin2011decentralized}, and more recently in \cite{furieri2019first}. For completeness, we report an explicit example in \preprintswitch{Appendix~B of our Arxiv report}{Appendix~\ref{ap:PL}}. Finally, notice that  $\mathcal{K}$ typically enforces a sparsity pattern for $\mathbf{K}$. Therefore, the QI property \eqref{eq:QI} can be checked without knowing the specific system dynamics, but only using the knowledge of the sparsity pattern of $\mathbf{CP}_{12}$ (see \cite{furieri2019unified} for example). This is a realistic assumption for dynamical systems that are distributed by nature.

\section{Learning the Globally Optimal Constrained Control Policy}
\label{sec:learning}

Here, we derive sample-complexity bounds for model-free learning of globally optimal distributed controllers for the problems identified in Section~\ref{sub:QI}. Our analysis technique is founded on recent zeroth-order optimization results \citep{malik2018derivative,fazel2018global}; we extend the derived bounds on the gradient estimates to include noise on the initial state, process noise and measurement noise. Furthermore, our analysis hinges on the observation that local gradient dominance is sufficient to guarantee the sample-complexity bounds in our framework, whereas \cite{malik2018derivative,fazel2018global} used a global one.
 The zeroth-order optimization literature is quite rich, see for instance the works of  \cite{balasubramanian2019zeroth,nesterov2017random,ghadimi2013stochastic} and references therein. The key idea of such algorithms is to sample noisy function values of $f$ generated by an \emph{oracle}, based on which an approximated gradient $\widehat{\nabla f}$ is estimated and standard gradient descent is applied to optimize over $z$.  While \cite{malik2018derivative} proposed an analysis for two-point evaluation oracles that allow for tighter sample-complexity bounds, we notice that in many control applications one cannot control or predict the noise affecting each separate measurement. 
We 
 will thus focus  on the 
  one-point evaluation oracle setup according to the Algorithm~\ref{algo} below.
\begin{algorithm}[h]
	\caption{ Model-free learning of distributed controllers}
	\label{algo}
	\begin{algorithmic}[1]
		\STATE Input: $z_0$, number of iterations $T$, stepsize $\eta>0$ and smoothing radius $r>0$.
		\FOR{$i = 0, \ldots, T-1$}
		\STATE Sample $u\sim \text{Unif}(\mathbb{S}_r)$, let nature ``choose'' disturbances $\delta_0\sim \mathcal{D}_{\delta_0}$, $w_t\sim \mathcal{D}_w$ for all $t=0,\ldots,N-1$, $v_t\sim \mathcal{D}_{v}$ for all $t=0,\ldots,N$.
		\STATE Apply  $\hat{\mathbf{u}}=\text{vec}^{-1}[P(z_i+u)]\hat{\mathbf{y}}$ iteratively using (\ref{eq:sys_disc}) and  store the resulting trajectories $\hat{\mathbf{y}},\hat{\mathbf{u}}.$
		\STATE Compute $\hat{f}=\hat{\mathbf{y}}^\mathsf{T}\text{blkdg}(M_0,\ldots,M_N)\hat{\mathbf{y}}+\hat{\mathbf{u}}^\mathsf{T}\text{blkdg}(R_0,\ldots R_{N-1})\hat{\mathbf{u}}$ and  $\widehat{\nabla f}=\hat{f}\frac{d}{r^2}u$.
		\STATE $z_{i+1}\leftarrow z_i-\eta \widehat{\nabla f}$.
		\ENDFOR
		\RETURN $\mathbf{K}_T= \text{vec}^{-1}(Pz_T)$.
	\end{algorithmic}
\end{algorithm}

In Algorithm~\ref{algo}, the observed cost $\hat{f}$ can be regarded as the output of a one-point evaluation oracle. Indeed, we have $\mathbb{E}_{\mathbf{w},\mathbf{v}}[\hat{f}]=f(z_i+u)$ by definition, and each observation $\hat{f}$ is affected by a different noise sequence. The value $\widehat{\nabla f}$ can be interpreted  as a noisy estimate\footnote{Technically, 
$\mathbb{E}[\widehat{\nabla f}] = \nabla f_r(z_i)$,
where $f_r(z_i)=\mathbb{E}_{u}[f(z_i+u)]$, with $u$ is  taken uniformly at random over $\mathbb{S}_r$; see, e.g., \cite[Lemma~6]{malik2018derivative} for details.} of the  gradient $\nabla f(z_i)$.  We now turn to the convergence analysis.

\subsection{Sample-complexity bounds}
Our sample-complexity analysis holds under two main assumptions on the function $f$.

\textbf{Assumption 1.~} For any $\delta>0$ and initial value $z_0 \in \mathbb{R}^d$, the sublevel set $\mathcal{G}_{10\delta^{-1}}$ of $f$ is compact.

\textbf{Assumption 2.~} For any $\delta>0$, the function $f$ admits a local gradient dominance constant $\mu_\delta$ over 
$\mathcal{G}_{10\delta^{-1}}$ as per (\ref{eq:LGD}).

In Section~\ref{sub:QI} we have  provided our main result about verifying that both assumptions hold for 1) all QI control problems and 2) some instances of non-QI problems, therefore establishing a novel fundamental connection between distributed control and zeroth-order optimization. As is common in zeroth-order analysis, we also verify Lipschitzness  and smoothness of $f$ on its sublevel sets.

\begin{lemma}
\label{le:Lipschitzianity}
 Let $\delta>0$ and Assumption~1 hold. Then, there  exist $\rho_0>0$, and $L_\delta,M_\delta>0$,  such that
\begin{equation}
\label{eq:Lipschitzianity_explicit}
\begin{aligned}
&|f(z')-f(z)|\leq L_\delta \norm{z'-z}_2\,,~~\norm{\nabla f(z')-\nabla f(z)}_2\leq M_\delta \norm{z'-z}_2\,,
 \end{aligned}
\end{equation}
for every $z', z \in \mathcal{G}_{10\delta^{-1}}$
such that $\norm{z'-z}_2\leq \rho_0$.
\end{lemma}
\begin{proof}
We know that $f$ is a multivariate polynomial. Since $\mathcal{G}_{10\delta^{-1}}$ is compact, it suffices to note that 
$\nabla f$ is a vector of polynomials and that polynomials are bounded on any compact set. 
\end{proof}
We are now ready to present the sample-complexity result. Its proof is reported in \preprintswitch{Appendix~C of our Arxiv report \cite{furieri2019learning}}{Appendix~\ref{app:theorem}}.
\begin{theorem}
\label{th:convergence_analysis}
Let Assumptions $1$ and $2$ hold, and consider Algorithm~\ref{algo}. Let $\eta>0$ and $r>0$ be selected according to
\begin{align*}
&\eta \leq \min \left\{ \frac{\epsilon \mu_\delta \delta^3 r^2}{16000 M_\delta d^2D^2f(z_0)^2 },\frac{1}{2M_\delta},\frac{\rho_0 r \delta}{20dDf(z_0)}  \right\},\\
&r \leq  \hspace{-0.1cm}\min\hspace{-0.1cm} \left\{\frac{\min\left(\hspace{-0.1cm}\frac{1}{2M_\delta},\frac{\rho_0}{L_\delta}\right) \mu_\delta}{2M_\delta} \sqrt{\frac{\delta \epsilon}{40}},\frac{1}{2 M_\delta}\sqrt{\frac{\epsilon \mu_\delta \delta}{5}}, \rho_0,\frac{10 \delta^{-1}f(z_0)}{L_\delta}\right\}\,,
\end{align*}
where $\rho_0>0$, $\mu_\delta$ is the local gradient dominance constant of $f(z)$ 
associated with $\mathcal{G}_{10\delta^{-1}}$, $L_\delta,$ $M_\delta$ are the local Lipschitzness and smoothness constants described in Lemma~\ref{le:Lipschitzianity}, and
 $D=\max \left( \frac{W^2}{\lambda_\mathbf{w}},\frac{V^2}{\lambda_\mathbf{v}}\right)$,  with $W$ the value such that $\norm{\mathbf{w}}_2\leq W$ for all $\delta_0\sim\mathcal{D}_{\delta_0},~w_0,\ldots,w_{N-1}\sim \mathcal{D}_{w}$, $V$ the value such that   $\norm{\mathbf{v}}_2\leq V$ for all $v_0,\ldots,v_{N}\sim \mathcal{D}_{v}$, and  $\lambda_\mathbf{w}$ and $\lambda_\mathbf{v}$ are the minimum eigenvalues of $\mathbb{E}[\mathbf{w}\mathbf{w}^\mathsf{T}]$ and $\mathbb{E}[\mathbf{v}\mathbf{v}^\mathsf{T}]$ respectively. Then for any  $\epsilon>0$ and $0<\delta<1$ such that $\epsilon \log (\frac{4\Delta_0}{\delta\epsilon})\leq \frac{16}{\delta}\Delta_0$, running Algorithm~\ref{algo} with $T= \frac{4}{\eta \mu_\delta}\log (\frac{4\Delta_0}{\delta\epsilon})$ iterations yields a \emph{distributed} control policy $\mathbf{K}_T \in \mathcal{K}$ such that
\begin{equation*}
J(\mathbf{K}_T)-J^\star\leq \epsilon\,,
\end{equation*}
with probability greater than $1-\delta$. 
\end{theorem}

Theorem~\ref{th:convergence_analysis} yields a constant probability suboptimality guarantee based on the analysis technique of \cite{malik2018derivative}; there, 
the infinite-horizon LQR  problem with exact state measurements and no information structure was addressed. Our result extends this analysis as follows. First, we consider distributed control problems, given noisy 
output information and allow for inclusion of both noise on the initial state, process noise and measurement noise. We achieve this by bounding the variance of the gradient estimate in our \preprintswitch{Lemma~13, which is reported in Appendix~C of our Arxiv report \cite{furieri2019learning}}{Lemma~\ref{le:G}, which is reported in Appendix~\ref{app:theorem}}.   Second, we allow for a success probability $1-\delta$ for any $\delta>0$ and show how $\delta$ affects the sample-complexity. Last, we observe that a \emph{local} gradient dominance constant valid on $\mathcal{G}_{10\delta^{-1}}$ is sufficient for the analysis, whereas \cite{malik2018derivative} considered a global one. Nonetheless, we note that the finite-horizon framework enjoys a significant simplification because any control policy leads to a finite cost and the feasible region of control policies is always connected. 
Extension to infinite-horizon requires further work. 

Based on Theorem~\ref{th:convergence_analysis}, the model-free sample-complexity scales as  $\mathcal{O}\left(\frac{d^2}{\epsilon^2\delta^4}\log{\frac{1}{\epsilon\delta}}\right)$. For the standard centralized LQR problem, model-based methods (e.g. \cite{dean2017sample}) can enjoy a better scaling of $\mathcal{O}\left(\frac{d}{\epsilon^2}\log{\frac{1}{\delta}}\right)$, but extension of these methods to the general $LQ_\mathcal{K}$ is non-trivial due to non-existence 
of a convex reformulation in general. 
Interestingly, the scaling with respect to the suboptimality gap $\epsilon$ is practically unaffected despite using a model-free method.

\subsection{Experiments for distributed control}

\begin{figure}[htbp]
\floatconts
{fig:example2}
{\caption{\scriptsize In Figure~\ref{fig:a} we plotted 1) the  average number of steps over $10$ runs of Algorithm~\ref{algo} needed to achieve $7$ increasingly tight precision levels from $\epsilon = 0.2$ to $\epsilon = 0.02$ and 2) the sample-complexity $T$ predicted by Theorem~\ref{th:convergence_analysis}, when $\eta$ is scaled as $\eta=\mathcal{O}\left( \epsilon r^2\right)$ and $r$ is scaled as $r=\mathcal{O}( \sqrt{\epsilon})$. In Figure~\ref{fig:b}, we plotted the convergence behaviour, highlighting the maximum and minimum cost achieved at each iteration among the $10$ runs. }}
{%
\subfigure{%
\label{fig:a}
\includegraphics[width=0.47\textwidth]{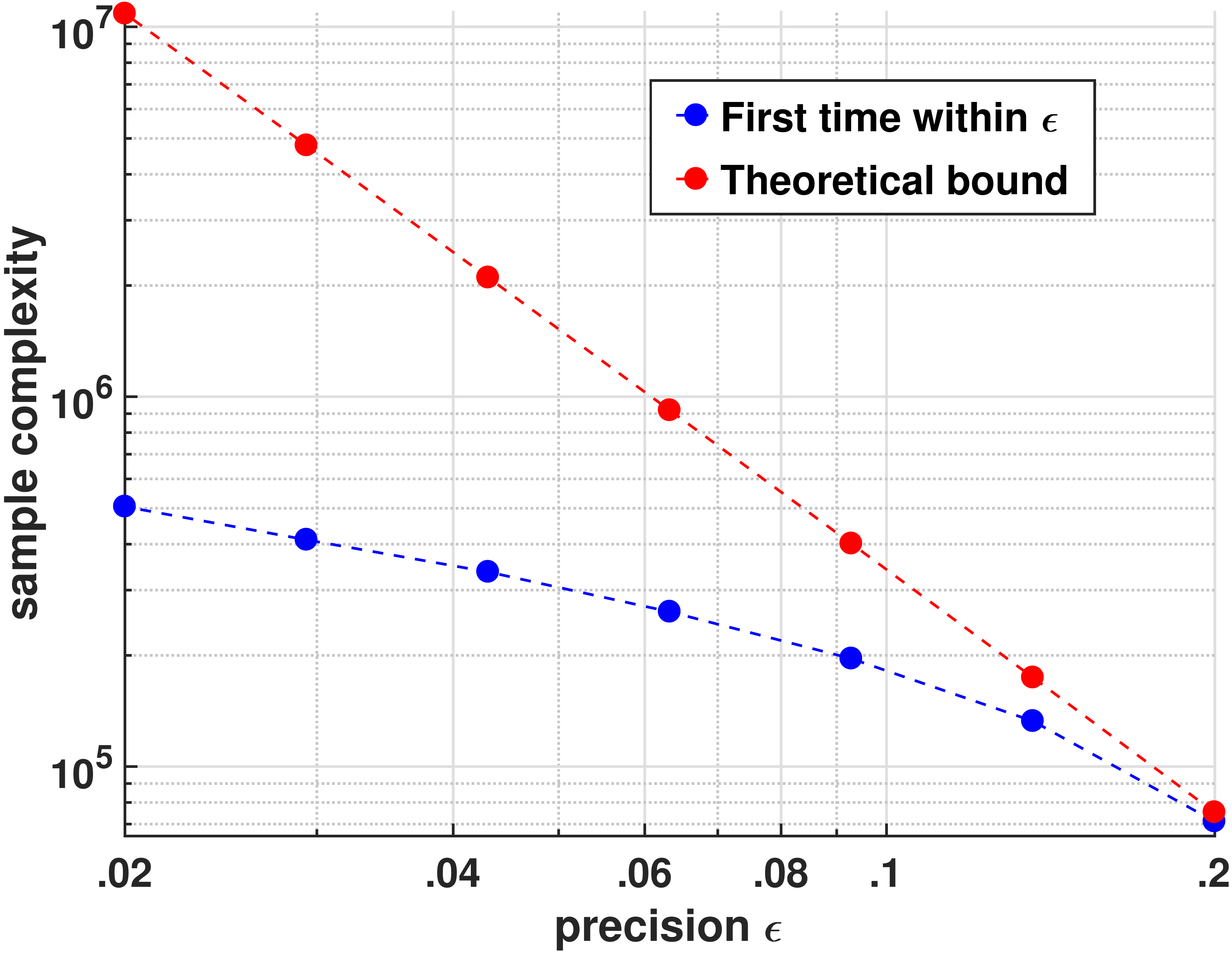}
}  \quad 
\subfigure{%
\label{fig:b}
\includegraphics[width=0.47\textwidth]{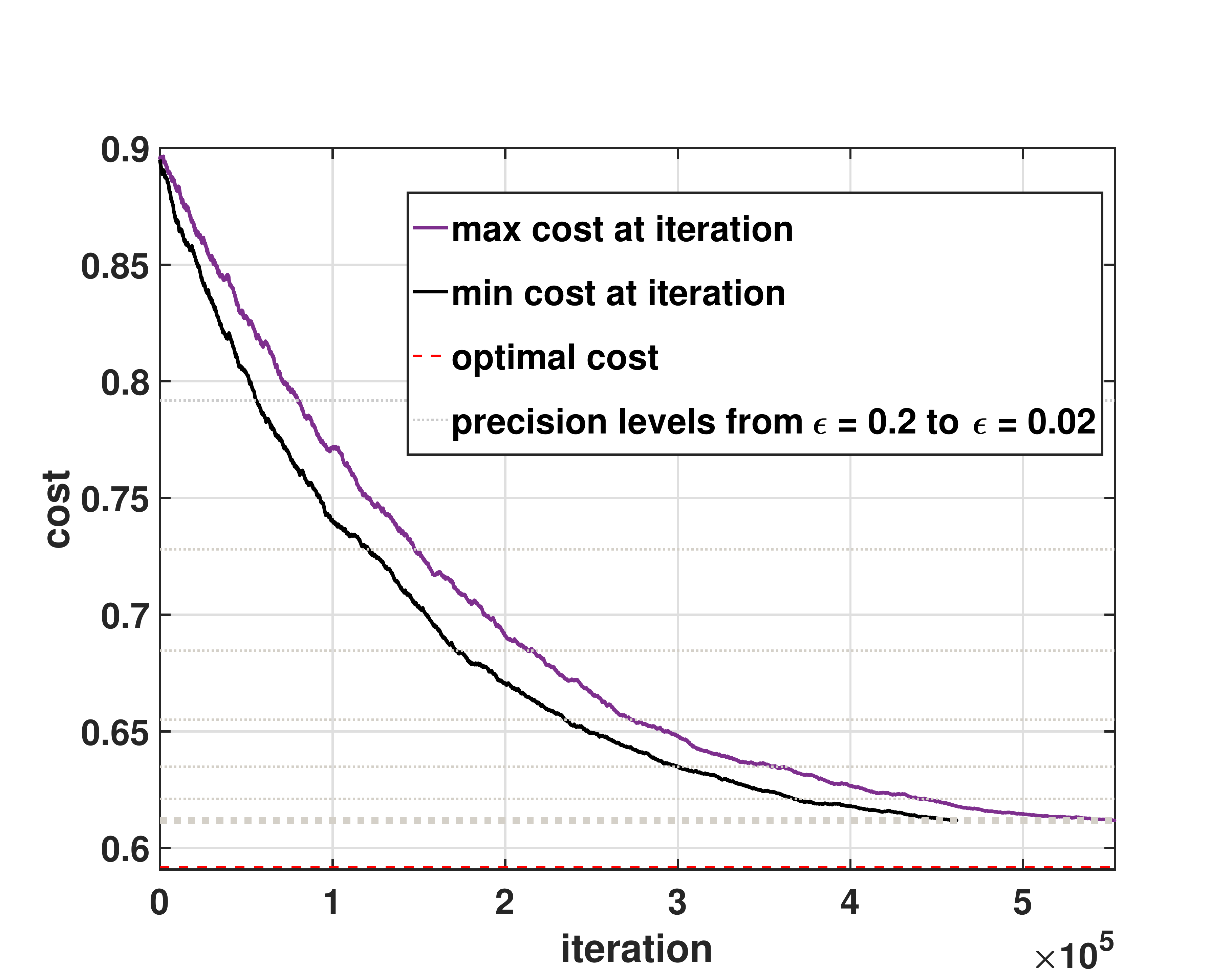}
}
}
\end{figure}
To validate our results, we considered problem $LQ_\mathcal{K}$ for $A_t=A$, $B_t=\begin{bmatrix}1&-1&0\end{bmatrix}^\mathsf{T}$, $C_t=I$ for  $t=0,1,2$,  $\mu_0=10^{-1}\times\begin{bmatrix}1&-1&1\end{bmatrix}^\mathsf{T}$,  with $\mathbf{K} \in \mathcal{K}=\text{Sparse}(\mathbf{S})$, where
\begin{equation*}
A=\begin{bmatrix}
1&0&-10\\-1&1&0\\0&0&1
\end{bmatrix}\,, \quad \mathbf{S}=\begin{bmatrix} 1 & 0&0 \\
1 & 1 & 0 \end{bmatrix}  
\otimes \begin{bmatrix}
1&0&0\end{bmatrix}\,.
\end{equation*}
Furthermore, we consider  additive initial state uncertainty uniformly distributed in the interval $[-10^{-2},10^{-2}]$, process noise and measurements noise $w_t$, $v_t$ uniformly distributed in the interval $[-10^{-3},10^{-3}]$ for every $t$.  The  cost function weights are chosen as $M_t=\frac{1}{4}I$ and $R_t=\frac{1}{4}I$ at each $t$. It is easy to verify that $\mathbf{KCP}_{12}\mathbf{K} \in \text{Sparse}(\mathbf{S})$ for any $\mathbf{K} \in \text{Sparse}(\mathbf{S})$; hence, $\mathcal{K}$ is QI with respect to $\mathbf{CP}_{12}$ and Theorem~\ref{th:convergence_analysis} holds.  Figure~\ref{fig:a}-\ref{fig:b} shows that the sample-complexity scales significantly better than the one predicted by Theorem~\ref{th:convergence_analysis} with respect to $\epsilon$, thus validating the corresponding bounds for this example.  Additional details and considerations on selecting $\eta$ and $r$ are reported in \preprintswitch{Appendix~D of our Arxiv report \cite{furieri2019learning}}{Appendix~\ref{app:experiments}}.




\section{Conclusions}
Motivated by the challenges of solving model-based distributed optimal control problems, we studied model-free policy learning subject to 
subspace constraints. By drawing a novel connection between gradient dominance and QI, we derived sample-complexity bounds on learning the globally optimal distributed controller for a class of problems including QI problems and other instances; for these, the available model-based learning techniques might not converge to a global optimum.
One exciting future direction is to extend
these results to infinite-horizon, by bridging the gap between dynamical controller synthesis and a gradient-descent landscape. We also envision including safety constraints. Furthermore, significantly sharpening our sample-complexity bounds might be possible with potentially more refined analysis. 


 \bibliography{references2}
\preprintswitch{}{
\newpage

\appendix


\section{Properties of the cost function}
\label{ap:propertiesJ}
In this appendix, we present a useful lemma that summarizes important properties of the cost function $J(\mathbf{K})$, which we use for our proofs. The first three points of Lemma~\ref{le:properties} provide the explicit expression of $J(\mathbf{K})$ as a function of $\mathbf{K}$ and its reformulation as a strongly convex function after applying an appropriate change of variables. Note that here we are still not enforcing an information structure. The fourth point of Lemma~\ref{le:properties} plays a key role in proving the local gradient dominance property in Theorem~\ref{th:PL2}. 

Results similar to Lemma~\ref{le:properties} are also presented  by \cite{furieri2019first}; however, \cite{furieri2019first} considered a slightly different cost function where the state trajectories are penalized directly. In the black-box setting considered here, it is not possible to observe the states, and hence it is only relevant to penalize the output trajectories in the cost. 

\begin{lemma}
\label{le:properties}
Let us define the operator $H:\mathbb{R}^{mN \times p(N+1)} \rightarrow \mathbb{R}^{mN \times p(N+1)}$ as
\begin{align}
\label{eq:H_operator}
&H(\mathbf{Q})=(I+\mathbf{Q}\mathbf{CP}_{12})^{-1}\mathbf{Q}\,.
\end{align}

The following facts hold.

\begin{enumerate}
\item The cost function $J(\mathbf{K})$ defined in \eqref{eq:cost} admits the following expression:
\begin{align}
J(\mathbf{K})&=\mathbb{E}_{\mathbf{w,v}}[\mathbf{y}^\mathsf{T}\mathbf{M}\mathbf{y}+\mathbf{u}^\mathsf{T}\mathbf{R}\mathbf{u}]\nonumber\\
&=\norm{\mathbf{M}^{\frac{1}{2}}\mathbf{C}(I-\mathbf{P}_{12}\mathbf{KC})^{-1}\mathbf{P}_{11}\mathbf{\Sigma}_w^{\frac{1}{2}}}_F^2+\norm{\mathbf{M}^{\frac{1}{2}}(I-\mathbf{C}\mathbf{P}_{12}\mathbf{K})^{-1}\mathbf{\Sigma}_v^{\frac{1}{2}}}_F^2\nonumber\\
&+\norm{\mathbf{R}^{\frac{1}{2}}\mathbf{K}(I-\mathbf{C}\mathbf{P}_{12}\mathbf{K})^{-1}\mathbf{C}\mathbf{P}_{11}\mathbf{\Sigma}_w^{\frac{1}{2}}}_F^2+\norm{\mathbf{R}^{\frac{1}{2}}\mathbf{K}(I-\mathbf{C}\mathbf{P}_{12}\mathbf{K})^{-1}\mathbf{\Sigma}_v^{\frac{1}{2}}}_F^2\label{eq:cost_K}\\
&+\norm{\mathbf{M}^{\frac{1}{2}}\mathbf{C}(I-\mathbf{P}_{12}\mathbf{KC})^{-1}\mathbf{P}_{11}\bm{\mu}_w}_2^2+\norm{\mathbf{R}^{\frac{1}{2}}\mathbf{K}(I-\mathbf{C}\mathbf{P}_{12}\mathbf{K})^{-1}\mathbf{C}\mathbf{P}_{11}\bm{\mu}_w}_2^2\,,\nonumber
\end{align}
where $\mathbf{M}=\text{\emph{blkdg}}(M_0,M_1,\ldots,M_N)$, $\mathbf{R}=\text{\emph{blkdg}}(R_0,\ldots R_{N-1})$, $\mathbf{\Sigma}_w=\text{\emph{blkdg}}(\Sigma_{\delta_0},I_{N-1} \otimes \Sigma_w)$,  $\mathbf{\Sigma}_v=I_N \otimes \Sigma_v$,  $\bm{\mu}_w=\begin{bmatrix}\mu_0^\mathsf{T}&0&\ldots&0\end{bmatrix}^\mathsf{T}$. 

\item Using~\eqref{eq:H_operator}, we have 
\begin{align}
J(H(\mathbf{Q}))&=\norm{\mathbf{M}^{\frac{1}{2}}\mathbf{C}(I+\mathbf{P}_{12}\mathbf{QC})\mathbf{P}_{11}\mathbf{\Sigma}_w^{\frac{1}{2}}}_F^2+\norm{\mathbf{M}^{\frac{1}{2}}(I+\mathbf{CP}_{12}\mathbf{Q})\mathbf{\Sigma}_v^{\frac{1}{2}}}_F^2 \nonumber\\
&+\norm{\mathbf{R}^{\frac{1}{2}}\mathbf{Q}\mathbf{C}\mathbf{P}_{11}\mathbf{\Sigma}_w^{\frac{1}{2}}}_F^2
+\norm{\mathbf{R}^{\frac{1}{2}}\mathbf{Q}\mathbf{\Sigma}_v^{\frac{1}{2}}}_F^2\label{eq:cost_Q}\\
&+\norm{\mathbf{R}^{\frac{1}{2}}\mathbf{Q}\mathbf{C}\mathbf{P}_{11}\bm{\mu}_w}_2^2+\norm{\mathbf{M}^{\frac{1}{2}}\mathbf{C}(I+\mathbf{P}_{12}\mathbf{QC})\mathbf{P}_{11}\bm{\mu}_w}_2^2\nonumber\,.
\end{align}
\item The function $J(H(\mathbf{Q}))$ is strongly convex in $\mathbf{Q}$.
\item The sublevel sets of $J(\mathbf{K})$ are compact.
\end{enumerate}
\end{lemma}
\begin{proof}
\begin{enumerate}
\item
Substituting~\eqref{eq:control_input_def} into~\eqref{eq:system_compact}, we have 
$$
    \mathbf{x}=(I-\mathbf{P}_{12}\mathbf{KC})^{-1}(\mathbf{P}_{11}\mathbf{w}+\mathbf{P}_{12}\mathbf{Kv}).    
$$
Then, it is not difficult to derive the closed-loop dynamics:
\begin{equation}\label{eq:closed_loop_K}
\begin{bmatrix}
\mathbf{y} \\
\mathbf{u}
\end{bmatrix} = \begin{bmatrix}
\mathbf{C}(I-\mathbf{P}_{12}\mathbf{KC})^{-1}\mathbf{P}_{11} & (I-\mathbf{CP}_{12}\mathbf{K})^{-1} \\
\mathbf{KC}(I-\mathbf{P}_{12}\mathbf{KC})^{-1}\mathbf{P}_{11}& \mathbf{K}(I-\mathbf{CP}_{12}\mathbf{K})^{-1}
\end{bmatrix}\begin{bmatrix}
\mathbf{w} \\
\mathbf{v}
\end{bmatrix}.
\end{equation}
For notational simplicity, we define 
$$
    \begin{bmatrix}
        \Phi_{yw} & \Phi_{yv} \\
        \Phi_{uw} & \Phi_{uv}
    \end{bmatrix}:= \begin{bmatrix}
\mathbf{C}(I-\mathbf{P}_{12}\mathbf{KC})^{-1}\mathbf{P}_{11} & (I-\mathbf{CP}_{12}\mathbf{K})^{-1} \\
\mathbf{KC}(I-\mathbf{P}_{12}\mathbf{KC})^{-1}\mathbf{P}_{11}& \mathbf{K}(I-\mathbf{CP}_{12}\mathbf{K})^{-1}
\end{bmatrix}.
$$
Then, we have 
$$
    \begin{aligned}
    \mathbf{y}^\tr\mathbf{M}\mathbf{y} &= \left(\Phi_{yw}\mathbf{w} + \Phi_{yv}\mathbf{v}\right)^\tr\mathbf{M}\left(\Phi_{yw}\mathbf{w} + \Phi_{yv}\mathbf{v}\right),\\
    \mathbf{u}^\tr\mathbf{R}\mathbf{u} &= \left(\Phi_{uw}\mathbf{w} + \Phi_{uv}\mathbf{v}\right)^\tr\mathbf{R}\left(\Phi_{uw}\mathbf{w} + \Phi_{uv}\mathbf{v}\right).
    \end{aligned}
$$

Observe that $\mathbb{E}_{\mathbf{w},\mathbf{v}}[\mathbf{w}^\mathsf{T}X\mathbf{v}]=0$ for any matrix $X$ of appropriate dimensions because $\mathbf{w}$ and $\mathbf{v}$ are independent, and the mean of $\mathbf{v}$ is $0$. Then, we have
\begin{equation}
\begin{aligned}
\label{eq:Phi_operators}
J(\mathbf{K})=&\;\mathbb{E}_{\mathbf{w}}\left[\mathbf{w}^\mathsf{T}\left(\Phi_{yw}^\tr \mathbf{M} \Phi_{yw} + \Phi_{uw}^\tr \mathbf{R} \Phi_{uw} \right)\mathbf{w}\right]  \\
&+ \mathbb{E}_{\mathbf{v}}\left[\mathbf{v}^\mathsf{T}\left(\Phi_{yv}^\tr \mathbf{M} \Phi_{yv} + \Phi_{uv}^\tr \mathbf{R} \Phi_{uv} \right)\mathbf{v}\right].
\end{aligned}
\end{equation}
Also, for any $X$ of compatible dimensions, we have   \begin{align*}
  &\mathbb{E}_{\mathbf{w}}\left(\mathbf{w}^\mathsf{T}X\mathbf{w}\right)=\text{Trace}(X\mathbf{\Sigma}_w)+\bm{\mu}_w^\mathsf{T}X\bm{\mu}_w\,,\\
  &\mathbb{E}_{\mathbf{v}}\left(\mathbf{v}^\mathsf{T}X\mathbf{v}\right)=\text{Trace}(X\mathbf{\Sigma}_v)\,.
  \end{align*}
  Recall that $\norm{X}_F^2=\text{Trace}(X^\mathsf{T}X)$, and it is straightforward to show that~\eqref{eq:Phi_operators} is the same with~\eqref{eq:cost_K}. 
  
  \item Using the change of variables in~\eqref{eq:H_operator}, \emph{i.e.},
  $$
    \mathbf{K} = H(\mathbf{Q}) := (I+\mathbf{Q}\mathbf{C}\mathbf{P}_{12})^{-1}\mathbf{Q},
  $$
  we can verify that 
  $$
    \mathbf{Q} = H^{-1}(\mathbf{K}) = \mathbf{K}(I - \mathbf{C}\mathbf{P}_{12}\mathbf{K})^{-1}.
  $$
  Furthermore, it is not difficult to derive that 
  $$
    \begin{bmatrix}
\mathbf{C}(I-\mathbf{P}_{12}\mathbf{KC})^{-1}\mathbf{P}_{11} & (I-\mathbf{CP}_{12}\mathbf{K})^{-1} \\
\mathbf{KC}(I-\mathbf{P}_{12}\mathbf{KC})^{-1}\mathbf{P}_{11}& \mathbf{K}(I-\mathbf{CP}_{12}\mathbf{K})^{-1}
\end{bmatrix} =  \begin{bmatrix}
\mathbf{C}(I + \mathbf{P}_{12}\mathbf{Q}\mathbf{C})\mathbf{P}_{11} & I + \mathbf{C}\mathbf{P}_{12}\mathbf{Q} \\
\mathbf{Q}\mathbf{C}\mathbf{P}_{11}& \mathbf{Q}
\end{bmatrix}. 
  $$
  Then, it is straightforward to verify the expression for $J(H(\mathbf{Q}))$ in~\eqref{eq:cost_Q}. 
  
  
  \item Since $J(H(\mathbf{Q}))$ is a sum of convex functions of $\mathbf{Q}$, it suffices that one of the addends is strongly convex. The addend $\norm{\mathbf{R}^{\frac{1}{2}}\mathbf{Q}\mathbf{\Sigma}_v^{\frac{1}{2}}}_F^2$ is strongly convex in $\mathbf{Q}$ due to $\mathbf{R},\mathbf{\Sigma}_v\succ 0$. 
  
  \item  Notice that any causal control policy $\mathbf{K}$ is such that the corresponding $\mathbf{Q}=H^{-1}(\mathbf{K})$ is causal, that is, $\mathbf{Q}$ is a block lower-triangular matrix.   By construction, $\mathbf{P}_{12}$ is a strictly block lower-triangular matrix and $\mathbf{C}$ is block-diagonal. Therefore, $\mathbf{QCP}_{12}$ is a strictly block lower-triangular matrix, and the matrix $\mathbf{QCP}_{12}$ is nilpotent. 
  
  It follows that 
  $$ \text{det}(I+\mathbf{QCP}_{12})=1. $$ 
  Thus,  $H(\mathbf{Q}) = (I+\mathbf{Q}\mathbf{C}\mathbf{P}_{12})^{-1}\mathbf{Q}$ is a matrix of polynomials in the entries of $\mathbf{Q}$. Then, $H(\cdot)$ is a bounded and continuous map, and the sublevel sets of $J(H(\mathbf{Q}))$ are compact in $\mathbf{Q}$ due to strong convexity. Finally, the sublevel sets of $J(\mathbf{K})$ are compact in $\mathbf{K}$ thanks to the boundedness and continuity of $H$.
  \end{enumerate}
\end{proof}

\section{QI background and local gradient dominance}
\label{ap:PL}


 The QI results of \cite{rotkowitz2006characterization,QIconvexity,furieri2019unified} revealed that a reformulation of $LQ_\mathcal{K}$ into an equivalent convex program is possible if and only if the  subspace constraint $\mathcal{K}$ satisfies (\ref{eq:QI}). In this case, $LQ_\mathcal{K}$ is equivalent to  
  \begin{equation}
\label{eq:convex_Q}
\min_{\mathbf{Q} \in \mathcal{K}}J(H(\mathbf{Q}))\,.
\end{equation}

To review this result, it is first convenient to define a vectorized version of the map $H(\cdot)$ and highlight some of its properties. Throughout this appendix, $P\in \mathbb{R}^{mpN(N+1)\times d}$ is a matrix such that its columns form an orthonormal basis of $\mathcal{K}$.
\begin{proposition}
\label{pr:bounded}
Let $h(\cdot):\mathbb{R}^d\rightarrow \mathbb{R}^d$ 
be the bijection defined as 
\begin{equation*}
h(q)=P^\mathsf{T}\text{\emph{vec}}\left(H\left(\text{\emph{vec}}^{-1}(Pq)\right)\right)\,.
\end{equation*}
We have that $h$ and $h^{-1}$ are continuous and bounded maps. 
\end{proposition}
\begin{proof}
The key of the proof is to show that for every $q \in \mathbb{R}^p$, $h(q)$ and $h^{-1}(q)$ are vectors of multivariate polynomials. This follows from the fact that $H(\mathbf{Q})$ is a matrix of polynomials as shown in the proof of Lemma~\ref{le:properties} and that the operator $\text{vec}(\cdot)$ simply stacks these polynomials into a vector. Note that $h^{-1}(z) = P^{\tr}\text{vec}\left(H^{-1}\left(\text{vec}^{-1}\left(Pz\right)\right)\right)$, and we can reason in the same way for $h^{-1}$. Since polynomials are continuous and bounded whenever $q$ is bounded, we conclude the proof.
\end{proof}
Then, the standard QI result can be stated as follows.
\begin{proposition}[QI]
\label{pr:QI}
Define $g:\mathbb{R}^d\rightarrow \mathbb{R}$ as $g(q)=f(h(q))$. The following two statements are equivalent.
\begin{enumerate}
\item The subspace $\mathcal{K}$ is QI with respect to $\mathbf{CP}_{12}$.
\item Problem $LQ_\mathcal{K}$ is equivalent to the \emph{unconstrained} strongly convex program
\begin{equation}
\label{eq:strongly_convex_program_vectorized}
\minimize_{q \in \mathbb{R}^d} g(q)\,.
\end{equation}
\end{enumerate}
\end{proposition}
\begin{proof}
 First, we have  $PP^\mathsf{T}=I$  since $P$ is the orthonormal basis matrix of the subspace $\mathcal{K}$. Then, by the definitions of $h(\cdot)$, $f(\cdot)$ and $g(\cdot)$, we compute
 \begin{align}
g(q)&=f\left(P^\mathsf{T}\text{vec}\left(H\left(\text{vec}^{-1}\left(Pq\right)\right)\right)\right)\nonumber\\
&=J\left(\text{vec}^{-1}\left(PP^\mathsf{T}\text{vec}\left(H\left(\text{vec}^{-1}\left(Pq\right)\right)\right) \right)\right)\nonumber\\
&=J\left(H(\text{vec}^{-1}(Pq))\right)\,, \label{eq:reformulated_q}
\end{align}
where we applied the fact that $f(z)=J(\text{vec}^{-1}(Pz))$. It follows that problem (\ref{eq:strongly_convex_program_vectorized}) is equivalent to minimizing (\ref{eq:reformulated_q})  over $\mathbb{R}^d$. Furthermore, since $P$ has a basis of $\mathcal{K}$ as its columns, $Pq$ spans $\mathcal{K}$ when $q$ spans over $\mathbb{R}^d$. Hence, minimizing (\ref{eq:reformulated_q}) is equivalent to solving problem (\ref{eq:convex_Q}). 
\end{proof}


Notice that if the system dynamics (\ref{eq:sys_disc}) are unknown, the system-dependent mapping $h$ is also unknown and thus one cannot apply zeroth-order optimization on the strongly convex function $g(q)$ directly. Instead, the strong-convexity property of $g$ can be exploited  to derive existence of a local gradient dominance constant for the original cost function $f$ as per Theorem~\ref{th:PL2}. We now report the proof of Theorem~\ref{th:PL2}.

\subsection{Proof of Theorem~\ref{th:PL2}}
Since $g$ is strongly convex thanks to QI, than its sublevel sets are compact \citep{boyd2004convex}, and the sublevel sets of $f$ must also be compact due to the fact that $h$ is a continuous and bounded map by Proposition~\ref{pr:bounded}. In addition, $g(q)$ has a unique stationary point due to its strong convexity. Then, $f(z)$ also has a unique stationary point since the mapping $z = h(q)$ is invertible.

We now turn to proving the local gradient dominance property, that is for any $\delta>0$ there exists $\mu_\delta$ such that \begin{equation}
\label{eq:gradient_dominance}
\mu_\delta(f(z)-J^\star)\leq  \norm{\nabla f(z)}_2^2\,, \quad \forall z \in \mathcal{G}_{10\delta^{-1}}\,.
\end{equation}
Let the functions $g(\cdot)$, $f(\cdot)$, $h(\cdot)$ be defined as per Proposition~\ref{pr:QI}. Since $g$ is strongly convex, there exists $\mu >0$ such that 
\begin{equation*}
    2\mu (g(q)-J^\star) \leq  \norm{\nabla g(q)}_2^2\,, \quad \forall q \in \mathbb{R}^d\,,
\end{equation*}
where $\mu$ is the strong-convexity constant of $g$ \citep{nguyen2017stochastic}. For every $q \in \mathbb{R}^d$, let $z=h(q)$. Then, thanks to the QI property,  we have 
$$f(z)=g(q), \quad  \forall q \in \mathbb{R}^d. $$ 

Let $\mathbf{J}_{h}:\mathbb{R}^d \rightarrow \mathbb{R}^{d \times d}$ denote the Jacobian function of $h(\cdot)$. Applying the derivative chain-rule leads to
\begin{align*}
    2\mu(f(z)-J^\star) &\leq \norm{\nabla g\left(h^{-1}(z)\right)}_2^2\,,\\
    &= \norm{\mathbf{J}_{h}\left(h^{-1}(z)\right) \nabla f\left(h\left(h^{-1}(z)\right)\right)}_2^2\\
    &\leq \norm{\mathbf{J}_{h}\left(h^{-1}(z)\right)}_F^2 \norm{\nabla f(z)}_2^2\,, \quad \forall z \in \mathbb{R}^d\,.
\end{align*}

Now, observe the Jacobian $\mathbf{J}_{h}\left(h^{-1}(z)\right)$ is a matrix of multivariate polynomials for every $z\in \mathbb{R}^d$. As such, each entry of $\mathbf{J}_{h}\left(h^{-1}(z)\right)$ is bounded on any compact set.  For the compact set $\mathcal{G}_{10\delta^{-1}}$ we denote
\begin{equation*}
    \tau=  \text{sup}_{z \in \mathcal{G}_{10\delta^{-1}}} \norm{\mathbf{J}_{h}\left(h^{-1}(z)\right)}_F^2\,,
\end{equation*}
which is a bounded constant. By setting $\mu_\delta=\frac{2\mu}{\tau}$, we obtain \eqref{eq:gradient_dominance}.

\subsection{Example of locally gradient dominated $LQ_\mathcal{K}$}
For clarity, we report an instance of $LQ_\mathcal{K}$ which is locally gradient dominated on its sublevel sets. Let $A_t=B_t=C_t=M_t=R_t=1$ for all $t=0,1,2$, $\mathbf{\Sigma}_w=\mathbf{\Sigma}_v=I$, $\mu_0=0$ and let us not enforce any subspace constraint. We have that
\begin{equation*}
    \mathbf{K}=\begin{bmatrix}a&0&0\\b&c&0\end{bmatrix}\,,\quad \mathbf{Q}=\begin{bmatrix}a_q&0&0\\b_q&c_q&0\end{bmatrix}=H^{-1}(\mathbf{K})=\begin{bmatrix}a&0&0\\b+ac&c&0\end{bmatrix}\,.
\end{equation*}
where $a,b,c\in \mathbb{R}$ are our decision variables and we can denote $z=\begin{bmatrix}a&b&c\end{bmatrix}^\mathsf{T}$ and $z_q=\begin{bmatrix}a_q&b_q&c_q\end{bmatrix}^\mathsf{T}$. One can verify that the cost function has the following expression:
\begin{equation*}
    J(\mathbf{K})=f(z)=(b+c+ac)^2+(b+ac)^2+2a^2+2c^2\,.
\end{equation*}
The mapping $h$ is given by 
\begin{equation*}
    z=h(z_q)=\begin{bmatrix}a_q&b_q-a_qc_q&c_q\end{bmatrix}^\mathsf{T}\,,
\end{equation*}
and we have that
\begin{equation*}
    g(z_q)=(b_q+c_q)^2+2a_q^2+b_q^2+2c_q^2\,,
\end{equation*}
is strongly convex with constant $\tau=5-\sqrt{5}$. Hence, $g$ is globally gradient dominated with a constant of $10-2\sqrt{5}$. The jacobian is given by
\begin{equation*}
    \mathbf{J}_h(h^{-1}(z))=\begin{bmatrix}1&0&0\\-c&1&-a\\0&0&1\end{bmatrix}\,.
\end{equation*}
Following the proof of Theorem~\ref{th:PL2}, we have that
\begin{equation*}
    (10-2\sqrt{5})f(z)\leq \norm{\mathbf{J}_h(h^{-1}(z))}_F^2\norm{\nabla f(z)}_2^2\,, \quad \forall z \in \mathbb{R}^3\,.
\end{equation*}

Since $\norm{\mathbf{J}_h(h^{-1}(z))}_F$ is bounded on any compact set we conclude that the gradient dominance property holds locally on any compact set. However, since it is unbounded on the real space, we cannot conclude that a global gradient dominance constant exists for $f$. 

Notice that in \cite{fazel2018global,malik2018derivative} etc. \emph{global} gradient dominance on the set of stabilizing controllers can be proven because the feasible region is bounded. In this paper, the feasible region is unbounded and therefore we need an adapted notion of local gradient dominance. As we prove in Theorem~\ref{th:convergence_analysis}, a local gradient dominance constant valid on the sublevel set $\mathcal{G}_{10\delta^{-1}}$ is sufficient to prove convergence; we observe that such a constant would also be sufficient to prove convergence in \cite{fazel2018global,malik2018derivative} and therefore derive sharper sample-complexity bounds.

\subsection{Example of non-QI $LQ_\mathcal{K}$ with local gradient dominance constant}
Consider the system given by $$A_t=\begin{bmatrix}1&2\\-1&-3\end{bmatrix},~~B_t,C_t,M_t,R_t,\Sigma_w,\Sigma_{\delta_0},\Sigma_v=I\,,$$ for $t=0,1,2$, with $\mu_0=\begin{bmatrix}0&1\end{bmatrix}^\mathsf{T}$. Let $\mathcal{K}$ be the subspace of matrices in the form $$\mathbf{K}=\begin{bmatrix}z_1&0&0&0&0&0\\0&z_2&0&0&0&0\\0&0&z_1&0&0&0\\0&0&0&z_2&0&0\end{bmatrix}$$ for any $z_1,z_2 \in \mathbb{R}$. The corresponding problem $LQ_\mathcal{K}$ is non-QI, yet it admits compact sublevel sets and a gradient dominance constant.

Indeed, it is immediate to verify that $\mathbf{KCP}_{12}\mathbf{K} \not \in \mathcal{K}$ for $\mathbf{K} \in \mathcal{K}$ in general. Therefore, QI does not hold. The cost function is given by $f(z_1,z_2)=4z_1^4+8z_1^3+30z_1^2+18z_1z_2-36z_1+6z_2^4-42z_2^3+151z_2^2-222z_2+191$.  It is easy to check similar to \cite{furieri2019first} that the corresponding Hessian is positive-definite for every $z_1,z_2 \in \mathbb{R}$. Therefore, $f(z)$ is strongly convex and it admits compact sublevel sets and a global gradient dominance constant. 

\section{Proof of Theorem~\ref{th:convergence_analysis}}
\label{app:theorem}
The proof consists of two parts. First, we prove a more general theorem which extends \cite[Theorem~1]{malik2018derivative} to the case of arbitrary constant probabilities and only relies on existence of a local gradient dominance constant, as opposed to \cite[Theorem~1]{malik2018derivative} where a global one is used.   Second, we derive bounds on the noisy gradient estimation in Algorithm~\ref{algo} for $LQ_\mathcal{K}$, where we allow for the superposition of noise on the initial state, process noise and output measurement noise.

\subsection{A generalization of \cite[Theorem~1]{malik2018derivative}}

Consider Algorithm~\ref{algo_Malik} below. 

\begin{algorithm}[h]
	\caption{Stochastic Zeroth-Order Method}
	\label{algo_Malik}
	\begin{algorithmic}[1]
		\STATE Input: $z_0$, number of iterations $T$, stepsize $\eta>0$ and smoothing radius $r>0$.
		\FOR{$t = 0, \ldots, T-1$}
		\STATE Sample $\xi_t\sim \mathcal{D}$ and $u_t\in \mathbb{S}_r$ uniformly at random
		\STATE $\tilde{\nabla}\leftarrow V(z_t+u_t,\xi_t)\frac{d}{r^2}u_t$
		\STATE $z_{t+1}\leftarrow x_t-\eta \tilde{\nabla}$
		\ENDFOR
		\RETURN $z_T$
	\end{algorithmic}
\end{algorithm}

The goal is to prove the following convergence result about Algorithm~\ref{algo_Malik}.
\begin{theorem}
\label{theorem_generic}
Let $v:\mathbb{R}^d\rightarrow \mathbb{R}$ be continuously differentiable and defined as $v(z)=\mathbb{E}_{\xi\sim\mathcal{D}}[V(z,\xi)]$. Consider the iteration of Algorithm~\ref{algo_Malik}. Define the set
\begin{equation*}
\mathcal{G}_{10\delta^{-1}}=\{x|~v(z)-v(z^\star)\leq 10\delta^{-1}(v(z_0)-v(z^\star)) \}\,,
\end{equation*}
where $z^\star$ is a global minimum of $v$ and $0<\delta<1$. Assume that $v$ is $(L_\delta,\rho_0)$ locally Lipschitz and $(M_\delta,\rho_0)$ locally smooth at every $z \in \mathcal{G}_{10\delta^{-1}}$, 
in the sense that
\begin{align*}
&|v(z')-v(z)|\leq L_\delta \norm{z'-z}_2\,,\\
 &\norm{\nabla v(z')-\nabla v(z)}_2\leq M_\delta \norm{z'-z}_2\,,
\end{align*}
for every $z',z\in \mathcal{G}_{10\delta^{-1}}$ 
such that $\norm{z'-z}_2\leq \rho_0$. 
Also assume that $\mu_\delta>0$ is a local gradient dominance constant for $v$ valid on $\mathcal{G}_{10\delta^{-1}}$.  For $u\sim \text{\emph{Unif}}( \mathbb{S}_r)$, define 
\begin{equation*}
G_\infty=\sup_{z \in \mathcal{G}_{10\delta^{-1}}}\norm{\frac{d}{r^2}V(z+u,\xi)u}_2\,, \;
G_2=\frac{d^2}{r^4}\sup_{z \in \mathcal{G}_{10\delta^{-1}}}\mathbb{E}\left[\norm{V(z+u,\xi)u-\mathbb{E}\left[V(z+u,\xi)u|~u\right]}_2^2\right]\,.
\end{equation*}
 Finally, define $\Delta_0=v(z_0)-v(z^\star)$. Then, by choosing the stepsize $\eta$ and the smoothing radius $r$ in Algorithm~\ref{algo_Malik} according to
\begin{align*}
&\eta \leq \min \left\{ \frac{\epsilon \mu_\delta \delta}{40 M_\delta G_2},\frac{1}{2M_\delta},\frac{\rho_0}{G_\infty}  \right\},\\
&r \leq  \min \left\{\frac{\min\left(\frac{1}{2M_\delta},\frac{\rho_0}{L_\delta}\right) \mu_\delta}{2M_\delta} \sqrt{\frac{\delta\epsilon}{40}},\frac{1}{2 M_\delta}\sqrt{\frac{\delta\epsilon \mu_\delta}{5}}, \rho_0\right\}\,,
\end{align*}
we have that for any given $\epsilon>0$ small enough such that $\epsilon \log{\left( \frac{4\Delta_0}{\delta\epsilon}\right)}<\frac{16}{\delta}\Delta_0$, after $T = \frac{4}{\eta \mu} \log \left(\frac{4 \Delta_0}{\delta \epsilon}\right)$ steps the iterate $z_T$ of Algorithm~\ref{algo_Malik} satisfies the bound
\begin{equation*}
v(z_T)-v(z^\star)\leq \epsilon\,,
\end{equation*}
with probability greater than $1-\delta$.
\end{theorem}

We proceed with the proof of Theorem~\ref{th:convergence_analysis}. As the proof method is based on the analysis of  \citep[Theorem~1]{malik2018derivative}, we will focus  on those aspects that require adaptation and/or extension, while explicitly referring to \cite[Theorem~1]{malik2018derivative} where appropriate.

For each iterate $z_t$, we define $\Delta_t = v(z_t) - v(z^*)$. For $0<\delta<1$, let $\tau:=\min \left\{t|~\Delta_t>10\delta^{-1}\Delta_0 \right\}$ be the first time instant when the iterate exits $\mathcal{G}_{10\delta^{-1}}$. Similar to \cite[Theorem~1]{malik2018derivative}, the proof is based on proving the following proposition stronger than Theorem~\ref{theorem_generic}. 

\begin{proposition}
\label{prop:stronger}
Setting the parameters as per Theorem~\ref{theorem_generic}, we have
\begin{equation*}
\mathbb{E}[\Delta_T 1_{\tau >T}]\leq \epsilon \frac{ \delta}{2}\,,
\end{equation*}
and  the event $\tau$ occurs after time step $T$ with probability greater than $1-\frac{\delta}{2}$. 
\end{proposition}

Let us first verify that Proposition~\ref{prop:stronger} implies Theorem~\ref{theorem_generic}. We have by using the probability sum rule and the Markov inequality that
\begin{align*}
\mathbb{P}\{\Delta_T\geq \epsilon\} &\leq \mathbb{P}\{\Delta_T 1_{\tau>T} \geq \epsilon\}+\mathbb{P}\{1_{\tau\leq T}\}\\
&\leq \frac{1}{\epsilon} \mathbb{E}[\Delta_T 1_{\tau >T}]+\mathbb{P}\{1_{\tau \leq T}\}\\
&\leq \frac{\delta}{2}+\frac{\delta}{2} = \delta\,.
\end{align*}
This is exactly the claim of Theorem~\ref{th:convergence_analysis}. Let now $\mathbb{E}^t[\Delta_{t+1}]$ denote expectation conditioned on all the randomness up to time $t$.  To prove Proposition~\ref{prop:stronger} we use the following Lemma. 

\begin{lemma}
\label{le:lemma5}
Given a function with the properties stated in Theorem~\ref{theorem_generic}, suppose  Algorithm~\ref{algo_Malik} is run with $r$ and $\eta$ chosen as per the statement of Theorem~\ref{theorem_generic}. Then, for any $t \in \mathbb{N}$ such that $z_t \in \mathcal{G}_{10\delta^{-1}}$ we have
\begin{equation}
\label{eq:Lemma5}
\mathbb{E}^t[\Delta_{t+1}]\leq \left(1-\frac{\eta \mu_\delta}{4}\right) \Delta_t+\frac{M_\delta \eta^2}{2}G_2+\eta \mu_\delta \frac{\epsilon \delta}{20}\,.
\end{equation}
\end{lemma}
\begin{proof}
The first part is to use the local smoothness and a local gradient dominance constant to derive
\begin{equation}
\label{eq:bound_crude_E}
\mathbb{E}^t[\Delta_{t+1} - \Delta_{t}] \leq -\frac{\eta \mu_\delta}{2}\Delta_t+\frac{\eta \mu_\delta}{4}\Delta_t+4\frac{\eta M_\delta^2 r^2}{\mu_\delta \theta_\delta^2}+\frac{M_\delta \eta^2}{2}G_2+M_\delta^3\eta^2 r^2\,,
\end{equation}
where we set $\theta_\delta=\min\left(\frac{1}{2M_\delta},\frac{\rho_0}{L_\delta}\right)$. The proof of this fact follows exactly \cite[Page~21]{malik2018derivative}, so we do not report it in full here. The single difference to notice is that we do not require a global gradient dominance cost, but one valid on $\mathcal{G}_{10\delta^{-1}}$. This is because the gradient dominance property is only applied to bound $\norm{\Delta v(z_t)}_2^2$, where $z_t \in \mathcal{G}_{10\delta^{-1}}$.

Next, we apply the bounds on $\eta$ and $r$. Using $r\leq \frac{\theta_\delta \mu_\delta}{2M_\delta}\sqrt{\frac{\epsilon \delta}{40}}$ we obtain
\begin{equation*}
4\frac{\eta M_\delta^2 r^2}{\mu_\delta \theta_\delta^2} \leq \frac{\eta \mu_\delta \epsilon \delta}{40}\,.
\end{equation*}
Using $\eta\leq \frac{1}{2 M_\delta}$ and $r\leq\frac{1}{2 M_\delta}\sqrt{\frac{\delta\epsilon \mu_{\delta}}{5}}$ we obtain
\begin{equation*}
M_\delta^3 \eta^2 r^2 \leq \frac{\eta \mu_\delta \epsilon \delta}{40}\,.
\end{equation*}
By plugging these bounds into \eqref{eq:bound_crude_E} and rearranging the equation we complete the proof.
\end{proof}
We proceed with the proof of Proposition~\ref{prop:stronger}. In the case where $\tau >T$, we can directly apply Lemma~\ref{le:lemma5} to bound the quantity $\mathbb{E}^t[\Delta_{t+1}]$. In the case where $\tau \leq T$ we have $\mathbb{E}^t[\Delta_{t+1}]1_{\tau >t}=0$ by definition. Combining these two cases and using the fact that $\eta\leq \frac{\epsilon \mu_\delta \delta}{40 M_\delta G_2}$ we have that
\begin{align*}
\mathbb{E}[\Delta_{t+1}1_{\tau>t+1}]&\leq \left(1-\frac{\eta \mu_\delta}{4}\right)^{t+1}\Delta_0+\left(\frac{M_\delta \eta^2}{2}G_2+\eta \mu_\delta \frac{\epsilon \delta}{20}\right)\sum_{i=0}^t\left(1-\frac{\eta \mu_\delta}{4}\right)^i\\
& \leq \left(1-\frac{\eta \mu_\delta}{4}\right)^{t+1}\Delta_0+\frac{2M_\delta \eta G_2}{\mu_\delta}+4\epsilon\frac{\delta}{20}\,,\\
&\leq \left(1-\frac{\eta \mu_\delta}{4}\right)^{t+1}\Delta_0+\frac{\epsilon \delta}{4}\,.
\end{align*}
Now set $t+1=T$. We want to ensure that
\begin{equation*}
\left(1-\frac{\eta \mu_\delta}{4}\right)^{T}\Delta_0 +\frac{\epsilon \delta}{4}\leq \frac{\epsilon \delta}{2}\,. 
\end{equation*}
 It can be verified that the above holds for  $T= \frac{4}{\eta \mu_\delta}\log{\left(\frac{4\Delta_0}{\delta \epsilon}\right)}$.  We conclude that with the parameters chosen as per the statement of the Theorem we have $\mathbb{E}[\Delta_{T}1_{\tau > T}] \leq \frac{\epsilon \delta}{2}$.
 
 We now turn to establishing that the event $\tau\leq T$ happens with a probability lower than $\frac{\delta}{2}$.  Similar to \cite[Theorem~1]{malik2018derivative}, the key ingredient is to identify a random variable associated with the iterates which is a super-martingale. Then, we can exploit classical inequalities on the maximum value of this random value up to time $T$. For each $t=1,\ldots, T$,  we define the stopped process
 \begin{equation*}
 Y_t =\Delta_{\min(\tau,t)} +(T-t)\left(\frac{M_\delta \eta^2}{2}G_2+\frac{\eta \mu_\delta \epsilon \delta}{20}\right)\,,
 \end{equation*}
 where $\tau$ is the first time step when $z_t \not \in \mathcal{G}_{10\delta^{-1}}$. By performing simple substitutions, we utilize the same derivations of \cite{malik2018derivative} to obtain
 \begin{equation*}
 \mathbb{E}^t[Y_{t+1}]\leq Y_t\,,
 \end{equation*} 
 that is, the stopped process $Y_t$ is a super-martingale. 
 
 We can now apply Doob's inequality (see for instance \cite{durrett2019probability}) and by substituting the values for $T$ and $\eta$ and using the requirement that $\epsilon$ is small enough to satisfy $\epsilon \log{\left(\frac{4\Delta_0}{\delta \epsilon}\right)} \leq \frac{16}{\delta} \Delta_0$ we obtain:
 \begin{align*}
 \mathbb{P}\left\{\max_{t=1,\ldots,T}Y_t \geq 10\delta^{-1}\Delta_0\right\}&\leq \frac{\mathbb{E}[Y_0]}{10\delta^{-1}\Delta_0}\\
 &=\frac{1}{10\delta^{-1}\Delta_0}\left(\Delta_0+T\left(\frac{M_\delta \eta^2}{2}G_2+\frac{\eta \epsilon \mu_\delta \delta}{20}\right)\right)\\
 &\leq \frac{1}{10\delta^{-1}\Delta_0}\left(\Delta_0+\frac{2}{\mu_\delta} \log{\left(\frac{4\Delta_0}{\delta \epsilon}\right)M_\delta\eta G_2}+\log{\left(\frac{4\Delta_0}{\delta \epsilon}\right)}\frac{\epsilon \delta}{5}\right)\\
 & \leq \frac{1}{10\delta^{-1}\Delta_0}\left(\Delta_0+\log{\left(\frac{4\Delta_0}{\delta \epsilon}\right)}\frac{\epsilon \delta}{20}+\log{\left(\frac{4\Delta_0}{\delta \epsilon}\right)}\frac{\epsilon \delta}{5}\right)\\
 &\leq \frac{\delta}{10 \Delta_0}\left(\Delta_0+\frac{16}{20}\Delta_0+\frac{16}{5}\Delta_0\right)=\frac{\delta}{2}\,.
 \end{align*}
Proposition~\ref{prop:stronger} is now proved. As observed above, Theorem~\ref{theorem_generic} is also proved.

 \subsection{Bounds on $G_2$ and $G_\infty$}
 
Before bounding the quantities $G_2$ and $G_\infty$  for the case of solving problem $LQ_\mathcal{K}$ with Algorithm~\ref{algo}, we derive useful inequalities  as follows.

\begin{lemma}
\label{le:10JK0}
Let $r\leq \min \left(\frac{10\delta^{-1} f(z_0)}{L_\delta},\rho_0 \right)$ and $u$ such that $\norm{u}_2= r$. If $z \in \mathcal{G}_{10\delta^{-1}}$, then 
\begin{equation*}
f(z+u) \leq 20 \delta^{-1}f(z_0)\,.
\end{equation*}  
\end{lemma}
\begin{proof}
First, we have 
$$
    \begin{aligned}
        |f(z+u)-f(z))|&\leq L_\delta r  \\
        &\leq 10\delta^{-1}f(z_0),
\end{aligned}
$$ 
where the first inequality comes from the definition of $L_\delta$  in \eqref{eq:Lipschitzianity_explicit} and $r\leq \rho_0$,  and the second inequality comes from $r\leq\frac{10\delta^{-1} f(z_0)}{L_\delta}$. Second, we have 
$$
    f(z)\leq 10\delta^{-1}f(z_0)
$$
because $z \in \mathcal{G}_{10\delta^{-1}}$ and $f(z^\star)\geq 0$. Combining the two inequalities above, we conclude that 
$$
    f(z+u)\leq f(z)+10\delta^{-1}f(z_0)\leq 20\delta^{-1}f(z_0).
$$
\end{proof}
\begin{lemma}
\label{le:bound_D}
For every $z \in \mathcal{G}_{10\delta^{-1}}$,  $r\leq \frac{10\delta^{-1}f(z_0)}{L_\delta}$ and $u$ such that $\norm{u}_2=r$ we have 
\begin{equation*}
\hat{\mathbf{y}}^\mathsf{T}\mathbf{M}\hat{\mathbf{y}}+\hat{\mathbf{u}}^\mathsf{T}\mathbf{R}\hat{\mathbf{u}}\leq D f(z+u) \leq 20\delta^{-1}D f(z_0)\,,
\end{equation*}
for every realization of the bounded disturbances, where $D=\max \left( \frac{W^2}{\lambda_\mathbf{w}},\frac{V^2}{\lambda_\mathbf{v}}\right)$,  with $W$ the value such that $\norm{\mathbf{w}}_2\leq W$ for all $\delta_0\sim\mathcal{D}_{\delta_0},~w_0,\ldots,w_{N-1}\sim \mathcal{D}_{w}$, $V$ the value such that   $\norm{\mathbf{v}}_2\leq V$ for all $v_0,\ldots,v_{N}\sim \mathcal{D}_{v}$, $\lambda_\mathbf{w}$ and $\lambda_\mathbf{v}$ are the minimum eigenvalues of $\mathbb{E}[\mathbf{w}\mathbf{w}^\mathsf{T}]$ and $\mathbb{E}[\mathbf{v}\mathbf{v}^\mathsf{T}]$ respectively.
\end{lemma}
\begin{proof}
We first prove that
\begin{equation}
\label{eq:bounded_positivedefinite}
\hat{\mathbf{w}} \hat{\mathbf{w}}^\mathsf{T} \preceq \frac{W^2}{\lambda_\mathbf{w}}\mathbb{E}[\mathbf{w}\mathbf{w}^\mathsf{T}]\,,\quad \hat{\mathbf{v}} \hat{\mathbf{v}}^\mathsf{T} \preceq \frac{V^2}{\lambda_\mathbf{v}}\mathbb{E}[\mathbf{v}\mathbf{v}^\mathsf{T}]\,.
\end{equation} 
Since  $\mathbb{E}[\mathbf{w}\mathbf{w}^\mathsf{T}]$ is symmetric, we have that $v^\mathsf{T}\mathbb{E}[\mathbf{w}\mathbf{w}^\mathsf{T}]v\geq \lambda_\mathbf{w}\norm{v}_2^2$ for all $v \in \mathbb{R}^{n(N+1)}$. We also have by Cauchy-Schwarz that $v^\mathsf{T}\hat{\mathbf{w}}\hat{\mathbf{w}}^\mathsf{T}v \leq W^2\norm{v}_2^2$ for all $v \in \mathbb{R}^{n(N+1)}$. Combining these two inequalities we deduce that the matrix $W^2 \lambda_\mathbf{w}^{-1}\mathbb{E}[\mathbf{w}\mathbf{w}^\mathsf{T}]-\hat{\mathbf{w}}\hat{\mathbf{w}}^\mathsf{T}$ is positive-definite. The same reasoning holds for $\mathbb{E}[\mathbf{v}\mathbf{v}^\mathsf{T}]$ and hence \eqref{eq:bounded_positivedefinite} holds.

 Next, for any matrix $X$ of appropriate dimensions, \eqref{eq:bounded_positivedefinite} implies that
\begin{align*}
&\text{Trace}(X\hat{\mathbf{w}} \hat{\mathbf{w}}^\mathsf{T}X^\mathsf{T})=(X \hat{\mathbf{w}})^\mathsf{T}X\hat{\mathbf{w}}\leq \frac{W^2}{\lambda_\mathbf{w}} \mathbb{E}[\mathbf{w}^\mathsf{T}X^\mathsf{T}X \mathbf{w} ]\,.
\end{align*}
The same reasoning applies to $\hat{\mathbf{v}}$. Now let $\mathbf{K}=\text{vec}^{-1}(P(z+u)) \in \mathcal{K}$. Recall equation (\ref{eq:Phi_operators}). It follows that
\begin{align*}
&\hat{\mathbf{y}}^\mathsf{T}\mathbf{M}\mathbf{y}+\hat{\mathbf{u}}^\mathsf{T}\mathbf{R}\hat{\mathbf{u}}\\
&=\hat{\mathbf{w}}^\mathsf{T}\left(\Phi_{yw}^\tr \mathbf{M} \Phi_{yw} + \Phi_{uw}^\tr \mathbf{R} \Phi_{uw} \right)\hat{\mathbf{w}}+\hat{\mathbf{v}}^\mathsf{T}\left(\Phi_{yv}^\tr \mathbf{M} \Phi_{yv} + \Phi_{uv}^\tr \mathbf{R} \Phi_{uv} \right)\hat{\mathbf{v}}\\
&\leq\frac{W^2}{\lambda_\mathbf{w}}\mathbb{E}_{\mathbf{w}}[\mathbf{w}^\mathsf{T}\left(\Phi_{yw}^\tr \mathbf{M} \Phi_{yw} + \Phi_{uw}^\tr \mathbf{R} \Phi_{uw} \right)\mathbf{w}]+\frac{V^2}{\lambda_\mathbf{v}}\mathbb{E}_{\mathbf{v}}[\mathbf{v}^\mathsf{T}\left(\Phi_{yv}^\tr \mathbf{M} \Phi_{yv} + \Phi_{uv}^\tr \mathbf{R} \Phi_{uv} \right)\mathbf{v}]\\
&\leq \max \left( \frac{W^2}{\lambda_\mathbf{w}},\frac{V^2}{\lambda_\mathbf{v}}\right)f(z+u)\,.
\end{align*}

Combining the above inequality with Lemma~\ref{le:10JK0} completes the proof.
\end{proof}

It is now straightforward to bound $G_\infty$ and $G_2$.
\begin{lemma}
\label{le:G}
Let $\hat{f}=\hat{\mathbf{y}}^\mathsf{T}\mathbf{M}\mathbf{y}+\hat{\mathbf{u}}^\mathsf{T}\mathbf{R}\hat{\mathbf{u}}$ be the cost associated with the observed trajectories resulting from applying the control policy $\hat{\mathbf{u}}=\text{\emph{vec}}^{-1}\left(P(z+u)\right)\hat{\mathbf{y}}$. For every $z \in \mathcal{G}_{10\delta^{-1}}$,  $r\leq \frac{10\delta^{-1}f(z_0)}{L_\delta}$ and $u$ such that $\norm{u}_2=r$ we have 
\begin{align*}
&G_\infty\leq \frac{20\delta^{-1}dD}{r}f(z_0)\,, \quad
G_2\leq \left(\frac{20\delta^{-1}dD}{r}f(z_0)\right)^2\,.
\end{align*}
\end{lemma}
\begin{proof}
 We have for any disturbances realization that
\begin{align*}
G_\infty&=\sup_{z \in \mathcal{G}_{10\delta^{-1}}}\norm{\hat{f} u \frac{d}{r^2}}_2=\frac{d}{r}\sup_{z\in \mathcal{G}_{10\delta^{-1}}}\hat{f}\leq \frac{20\delta^{-1}dD}{r}f(z_0)\,,
\end{align*}
where the last inequality is an application of Lemma~\ref{le:bound_D}. Since it holds that $G_2\leq G_\infty^2$ by definition, we also obtain the bound for $G_2$.
\end{proof}

Finally, Theorem~\ref{th:convergence_analysis} is proven by simply substituting the bounds for $G_2$ and $G_\infty$ derived in Lemma~\ref{le:G} into Theorem~\ref{theorem_generic}.

\section{Details on Experiments}
\label{app:experiments}

For the control problem described above, it is easy to solve (\ref{eq:convex_Q}) with convex programming and verify that the optimal cost is given by $J(\mathbf{K}^\star)=f(z^\star)=J^\star=0.5918$, where 
\begin{equation*}
\mathbf{K}^\star=\begin{bmatrix}
2.7881&0&0&0&0&0&0&0&0\\
-0.2284&0&0&0.9833&0&0&0&0&0
\end{bmatrix}\,, \quad z^\star=\begin{bmatrix}
2.7881&-0.2284&0.9833
\end{bmatrix}^\mathsf{T}\,.
\end{equation*}
To test the model-free performance of Algorithm~\ref{algo} and validate the result of Theorem~1, we proceeded as follows. We first picked an initial control policy $z_0=z^\star-\begin{bmatrix}1&1&1\end{bmatrix}^\mathsf{T}$, which is such that $f(z_0)=0.8951$ and thus $\Delta_0=0.3033$. Selecting the stepsize is a notoriously delicate task, inherent to reinforcement learning approaches \cite[Chapter 6]{bertsekas2011dynamic}; for this reason, the values $\eta=0.0005$ and $r=0.1$ were selected by trial-and-error, until the satisfactory convergence behaviour of Figure~\ref{fig:b} was obtained. A rigorous validation for this choice is beyond the scope of the paper.  Then, we plotted 1) the  average number of steps over $10$ runs of Algorithm~\ref{algo} needed to achieve $7$ increasingly tight precision levels from $\epsilon = 0.2$ to $\epsilon = 0.02$ and 2) the sample-complexity $T$ predicted by Theorem~\ref{th:convergence_analysis} when $\eta$ is scaled as $\eta=\mathcal{O}\left( \epsilon r^2\right)$ and $r$ is scaled as $r=\mathcal{O}( \sqrt{\epsilon})$. We refer to Figure~\ref{fig:a} for the corresponding plots. We verified that for each precision level, by stopping the algorithm  exactly at the iterations $T$ shown in red in Figure~\ref{fig:a}, the corresponding $z_T$ was within the desired precision level $10/10$ of the runs.

}

\end{document}